\documentclass[fleqn,10pt]{wlscirep}
\usepackage[utf8]{inputenc}
\usepackage[T1]{fontenc}

\usepackage{amsmath}
\usepackage{amsthm}

\newtheorem{theorem}{Theorem}

\title{Reinforcement Quantum Annealing: A Quantum-Assisted Learning Automata Approach}

\author[1,*]{Ramin Ayanzadeh}
\author[1]{Milton Halem}
\author[1]{Tim Finin}

\affil[1]{Department of Computer Science and Electrical Engineering, University of Maryland, Baltimore County, Baltimore, MD 21250, United States}
\affil[*]{ayanzadeh@umbc.edu}


\begin{abstract}
We introduce the reinforcement quantum annealing (RQA) scheme in which an intelligent agent interacts with a quantum annealer that plays the stochastic environment role of learning automata and tries to iteratively find better Ising Hamiltonians for the given problem of interest. 
As a proof-of-concept, we propose a novel approach for reducing the NP-complete problem of Boolean satisfiability (SAT) to minimizing Ising Hamiltonians and show how to apply the  RQA for increasing the probability of finding the global optimum. 
Our experimental results on two different benchmark SAT problems (namely factoring pseudo-prime numbers and random SAT with phase transitions), using a D-Wave 2000Q quantum processor, demonstrated that RQA finds notably better solutions with fewer samples, compared to state-of-the-art techniques in the realm of quantum annealing.
\end{abstract}

\begin{document}
\flushbottom
\maketitle
\thispagestyle{empty}

\section*{Introduction}
Quantum artificial intelligence and quantum machine learning are emerging fields that leverage quantum information processing to address a certain types of problems, which are intractable in the realm of classical computing \cite{lamata2017basic,biamonte2017quantum,dunjko2018machine}. There are several models for the physical realization of quantum computers \cite{ladd2010quantum}. Among the quantum computing models, adiabatic quantum computers are currently more readily available at user sites due to recent advancements in commercializing programmable quantum annealers by D-Wave Systems \cite{johnson2011quantum}. 

Quantum annealing is a meta-heuristic that uses Quantum-mechanical fluctuations to address discrete optimization problems (i.e., combinatorial optimization problems), which are intractable in the realm of classical computing \cite{finnila1994quantum,kadowaki1998quantum,ohzeki2011quantum,nishimori2017exponential}. Quantum annealers are a type of adiabatic quantum computer that provides a hardware implementation for minimizing Hamiltonians whose ground states represent optimum solutions of the original problems of interest. The D-Wave quantum annealer is a programmable (and commercialized) Ising processing unit (IPU) that can minimize the  transverse Ising Hamiltonian or its equivalent quadratic unconstrained binary optimization (QUBO) form \cite{johnson2011quantum}. More precisely, the D-Wave quantum annealer receives coefficients of an Ising Hamiltonian (here $\mathbf{h}$ and $J$) as an executable quantum machine instruction (QMI) and returns the vector $\mathbf{z}$ that minimizes the following energy function: 
\begin{equation}
	\label{eqn:ising_energy}	
	E_{\mathrm{Ising}}{(\mathbf{z})} = \sum_{i=1}^{N}{\mathbf{h}_i\mathbf{z}_i} + \sum_{i=1}^{N}{\sum_{j=i+1}^{N}{J_{ij}\mathbf{z}_i\mathbf{z}_j}},
\end{equation}
where ${N}$ denotes the number of quantum bits (qubits) and $\mathbf{z}_i \in \{-1, +1\}$. To solve a problem on a D-Wave quantum annealer, one needs to define an Ising Hamiltonian, shown in Eq. \eqref{eqn:ising_energy}, whose ground state represents a solution for the original problem of interest \cite{ayanzadeh2019an}. 

The current generation of the D-Wave quantum annealer (the Chimera architecture) includes more than 2,000 qubits and about 6,000 couplers, while the next generation (the Pegasus topology of the Advantage) will include more than 5,000 qubits and about 40,000 couplers \cite{boothby2019next}. Recent studies have revealed the potential of quantum annealers (namely the D-Wave quantum processors) to address certain classes of real-world problems that are intractable in the realm of classical computing \cite{biswas2017nasa}—including, but not limited to, planning \cite{rieffel2015case}, scheduling \cite{venturelli2015quantum,tran2016hybrid}, discrete optimization problems \cite{bian2014discrete}, constraint satisfaction problems \cite{bian2016mapping}, Boolean satisfiability \cite{su2016quantum,bian2017solving}, matrix factorization \cite{o2018nonnegative}, cryptography \cite{peng2019factoring}, fault detection and system diagnosis \cite{perdomo2015quantum}, compressive sensing \cite{ayanzadeh2019quantum}, control of automated vehicles \cite{ohzeki2018control} and protein folding \cite{perdomo2012finding}. In addition, by sampling from high-dimensional probability distributions, one can use the D-Wave quantum annealers for many applications in artificial intelligence, machine learning and signal processing \cite{biswas2017nasa,adachi2015application,vinci2019path}. 

Beside all aforementioned applications, the D-Wave quantum annealer architecture has limitations that not only restrict the process of mapping problems into an executable QMI (namely the sparse connectivity) but also lower the quality of results. Reducing a problem to an Ising Hamiltonian that represents the solution of the problem in its ground state does not guarantee that executing the corresponding QMI on a D-Wave quantum processing unit (QPU) will attain the optimum solution. For a given QMI, the D-Wave QPU draws samples from a problem-dependent pseudo-Boltzmann distribution at cryogenic temperatures \cite{biswas2017nasa}. The energy value of samples from the D-Wave QPU follow a Gaussian distribution. Thus, when we increase the number of reads/samples, we expect that the average parameter in the corresponding Gaussian distribution to approach the ground state energy of the corresponding Ising Hamiltonian—i.e., the probability of finding the global minimum approaches one. There are several drawbacks, nevertheless, that prevent quantum annealers from attaining a global minimum—including, but not limited to, confined anneal time, coefficients’ range and precision limitations, noise, and decoherence.

From a quantum computing perspective, an adiabatic quantum computer needs to minimize a non-stoquastic Hamiltonian in order to be universal (which would make them equivalent to gate models); nevertheless, the D-Wave QPU minimizes an Ising Hamiltonian which is stoquastic \cite{nishimori2017exponential,vinci2017non}.
Since coupling every qubit to every other qubit in a quantum annealer is impractical, the D-Wave QPU has a sparse structure/architecture. Hence, we entangle multiple qubits to represent virtual qubits with higher connectivity. Chaining physical qubits substantially reduces the capacity of QPUs—e.g., 2,048 qubits in the Chimera architecture is equivalent to a clique of size 64.
It is possible to implicitly leverage the capacity of the current D-Wave QPUs \cite{okada2019improving}, albeit executing multiple QMIs. In addition, virtual qubits are vulnerable to breaking—the longer the chains, the higher the probability they break during the annealing process. Although we can remediate broken chains by applying postprocessing methods on classical computers (e.g., voting among the physical qubits on a chain), some chains break because they represent a state with lower energy.

The required anneal time in a quantum annealer to keep the process adiabatic has a reverse exponential relation to the energy gap between the ground state (global minimum) and the first excited state (a state right above the global minimum) \cite{nishimori2017exponential}. In the current generation of the D-Wave QPU, $\mathbf{h}_i \in [-2,+2]$ and $J_{ij} \in [-1,+1]$. Thus, one needs to scale the resulting Ising model, Eq. \eqref{eqn:ising_energy}, by dividing all coefficients with a large-enough positive number to satisfy the QPU hardware constraints. However, scaling the Ising model reduces the energy gap between the ground and the first excited states. As a result, the required annealing time can quickly exceed the maximum possible anneal time on a physical quantum annealer (for example 2,000 micro-seconds on the D-Wave QPUs) and makes the process diabatic, which exponentially reduces the probability of getting to the ground state \cite{nishimori2017exponential}.

The current generation of the D-Wave QPUs use 8–9 bits for representing coefficients in Eq. \eqref{eqn:ising_energy}. Hence, the D-Wave QPU truncates coefficients of a QMI prior to putting qubits in their superposition, which can result in  the Ising model having a different ground state—compared to the original QMI. Consequently, the D-Wave QPU may solve a different problem whose result is either infeasible or less accurate than the original problem of interest \cite{pudenz2015quantum,dorband2018extending}. Applying preprocessing techniques \cite{pelofske2019optimizing} and classical postprocessing heuristics \cite{dorband2018method} can remarkably enhance the performance of the D-Wave QPU; however, from a problem-solving viewpoint, the D-Wave quantum annealer cannot guarantee achieving a global optimum.

In this study, we view quantum annealers from two different perspectives simultaneously: (i) a meta-heuristic for solving discrete optimization problems that can find very high-quality solutions in near-constant time; and (ii) a physical process that naturally draws samples from a problem-dependent Boltzmann distribution in cryogenic temperatures. Unlike most current research in Quantum artificial intelligence that applies quantum computing models to hard AI problems, in  this paper, we explore how we might apply AI techniques to improve quantum information processing. 
This paper presents the Reinforcement Quantum Annealing (RQA) scheme that leverages the idea of learning automata \cite{narendra2012learning} to iteratively improve the quality of results, attained by the quantum annealers, and implicitly address the aforementioned issues. In this scheme, the agent maps the problem to an executable QMI and interacts with a quantum annealer that plays the role of a stochastic environment in learning automata. As a proof-of-concept, we first introduce a novel approach for reducing Boolean satisfiability (SAT) instances \cite{biere2009handbook} to minimizing Ising Hamiltonians and then demonstrate that adopting the proposed RQA scheme results in notably better solutions.

\section*{Method}
Assume that the given problem $\Pi$ that we aim to solve on a quantum annealer contains a finite set of constraints (components), denoted by $\pi_i$ for $i \in \{1, 2, \dots, M \}$, over the same variables as follows: 
\begin{equation}
	\label{eqn:problem_constraints}
	\Pi := \{\pi_i, \pi_2, \dots, \pi_M \},
\end{equation}
where ${M}$ indicates the number of constraints and our ultimate objective is to find a solution that addresses all  constraints. Let $H_i$ be an Ising Hamiltonian whoes ground state represents a solution for $\pi_i$ and $H_{\Pi}$ be the corresponding Ising Hamiltonian of $\Pi$ that all are acting on same spins (variables). In addition, let $E_{0}^{H_{\Pi}}$ be the ground state energy of $H_{\Pi}$ and $E_0^{H_i}$ be ground state energy of corresponding $H_i$. If there exists $\mathbf{z}$ that puts $H_i$ ($\forall{i}, i \in \{1, 2, \dots, M\}$) in their ground states (i.e., satisfies all constraints in $\Pi$), then $\mathbf{z}$ also puts $H_{\Pi}$ in its ground state \cite{mooney2019mapping}—in other words, 
\begin{equation} 
	\label{eqn:hamiltonians_equality}
	E_0^{H_{\Pi}} = E_0^{H_1} + E_0^{H_2} + \dots + E_0^{H_M}.
\end{equation}
Hence, we can represent $H_{\Pi}$ as follows:
\begin{equation} 
	\label{eqn:hamiltonian_composition}
	H_{\Pi} = H_1 + H_2 + \dots + H_M.
\end{equation}
This setting appears in a vast range of problem formulations—including, but not limited to SAT \cite{bian2017solving}, constraint satisfaction problems \cite{vyskocil2018simple,vyskovcil2019embedding}, planning and scheduling \cite{rieffel2015case,venturelli2015quantum,tran2016hybrid}, fault detection and diagnosis \cite{perdomo2015quantum,bian2016mapping}, and compressive sensing \cite{ayanzadeh2019quantum,mousavi2019survey}—specifically when we adopt the idea of penalty methods to reduce problems of interest to minimizing Ising Hamiltonians.

\begin{theorem}
	\label{thm:infinite_Ising}
	For any problem in class NP, there are infinite different Ising models whose ground states are all identical to the solution of the original problem. 
\end{theorem}
\begin{proof}
	According to Cook—Levin theorem, we can reduce any NP problem to minimizing an Ising Hamiltonian (which is also in the class NP) in polynomial-time \cite{garey2002computers}. Multiplying all coefficients of the Ising model by a positive non-zero real number results in a new Ising model whose ground state will be identical to the original Ising model. Since the number of positive real numbers are infinite, we can generate infinite different Ising models whose ground states represent the solution for the original problem of interest. 
\end{proof}

According to Theorem \ref{thm:infinite_Ising}, there are infinite different Ising Hamiltonians whose ground states all represent the solution of the original problem of interest—nevertheless owing to the range and precision limitations on the D-Wave QPUs, we have finite different Ising models for a problem. In theory, these different Ising models are equivalent to each other—i.e., an adiabatic annealing process always attains the ground state which is identical for all corresponding Ising models of a problem. In practice, however, each of these (theoretically) equivalent Ising models are analogous to a pseudo-Boltzman distribution whose parameters are different. Consequently, when we minimize the corresponding QMIs with a physical quantum annealer (like the D-Wave QPU), the probability of finding the global minimum for different  Ising Hamiltonians of a given problem varies from zero to one. As an example, an annealing process on a D-Wave QPU may become diabatic because the required anneal time exceeds the maximum possible anneal time (2,000 micro-seconds), which can substantially reduce the probability of finding the global minimum. Note that for a given Ising Hamiltonian, we cannot estimate the probability of finding the ground state prior to executing the corresponding QMI. 

We introduce a novel scheme, called reinforcement quantum annealing (RQA),  in which an intelligent agent interacts with a quantum annealer as a stochastic environment and tries to iteratively find better Ising Hamiltonians that sampling from their ground state(s), by quantum annealers, results in a better distribution—i.e., the probability of finding the global optimum is increased over the time.
To this end, we extend Eq. \eqref{eqn:hamiltonian_composition} as follows:
\begin{equation} 
	\label{eqn:RQA_ext_h}
	\widetilde{H}_{\Pi} = \widetilde{H}_{1} + \widetilde{H}_{2} + \dots + \widetilde{H}_{M},
\end{equation}
such that: 
\begin{equation} 
	\label{eqn:RQA_chi}
	\widetilde{H}_{i} = \chi \left({H_i, \boldsymbol{\rho}_i} \right),
\end{equation}	
where $\boldsymbol{\rho}_i \in \mathbb{R}$ denotes the impact (or influence) factor of $\pi_i$ (or $H_i$) and $\chi$ is a function that maps the input Hamiltonian to a different Hamiltonian which satisfies:
\begin{itemize} 
	\item any $\mathbf{z}$ that puts $H_i$ in its ground state also puts $\widetilde{H}_i$ in its ground state, and vice versa;
	\item if $\boldsymbol{\rho}_i^1 < \boldsymbol{\rho}_i^2$ then $\chi \left({H_i, \boldsymbol{\rho}_i^1}\right) \ge \chi \left({H_i, \boldsymbol{\rho}_i^2}\right)$.
\end{itemize}

Learning automata (LA) \cite{narendra2012learning} are adaptive decision-making models (i.e., type of reinforcement learning \cite{kaelbling1996reinforcement}) that try to maximize the accumulative reward when they are interacting with stochastic environments. In a similar manner to reinforcement learning, LA use Markov-decision-processes for representing the automaton-environment structure \cite{kaelbling1996reinforcement,narendra2012learning,sutton2018reinforcement}. In a learning automaton, the agent has a set of ${r}$ actions (denoted by ${\alpha} = \{ {\alpha}_1, {\alpha}_2, \dots, {\alpha}_r \}$) and each action has a corresponding probability (denoted by $\boldsymbol{p}_i$ and $\sum{\boldsymbol{p}_i}=1$). In each episode, the agent takes (applies) the action ${\alpha}_i$ (according to $\boldsymbol{p}$), and (correspondingly) the stochastic environment returns its feedback $\beta$ that specifies the performance evaluation of the action ${\alpha}_i$. The agent uses this feedback to learn from the environment and aims to take optimal actions over the time. For $\beta \in [0,1]$ —so-called S-Type learning automata—in episode ${t}$, if the agent takes the action ${\alpha}_i$ and receives the feedback $\beta^{t}$, we can update $\boldsymbol{p}$ as follows:
\begin{equation}
	\label{eqn:la_singletask_update}
	\boldsymbol{p}_j^{t+1} = 
	\begin{cases}
		\boldsymbol{p}_j^{t} -\theta_2 \left({1-\beta^{t}}\right) \boldsymbol{p}_j^{t} + \theta_1 \beta^{t} \left({1-\boldsymbol{p}_j}\right), \quad \quad  & i=j;\\
		\boldsymbol{p}_j^{t} + \theta_2 \left({1-\beta^{t}}\right) \left({\frac{1}{r-1}-\boldsymbol{p}_i^{t}}\right) - \theta_1 \beta^{t} \boldsymbol{p}_{j}, \quad \quad  & i \ne j,
	\end{cases}
\end{equation}
where $\beta=0$ represents the lowest action performance and $\beta=1$ represents the highest action performance, $\theta_1 , \theta_2 \in [0,1]$ are learning factors, and $i,j \in \{1, 2, \dots, r \}$ \cite{narendra2012learning}.  

We extend learning automata to allow the agent to take multiple actions in each episode. Let $\hat{{\alpha}}^t \subset {\alpha}$ denotes the set (list) of actions that the agents takes in episode ${t}$. We can extend the Eq. \eqref{eqn:la_singletask_update} as follows:
\begin{equation}
	\label{eqn:la_multitask_update}
	\boldsymbol{p}_j^{t+1} = 
	\begin{cases}
		\boldsymbol{p}_j^{t} -\theta_2 \left({1-\beta^{t}}\right) \boldsymbol{p}_j^{t} + \theta_1 \beta^{t} \left({1-\boldsymbol{p}_j}\right), \quad \quad  & {\alpha}_j \in \hat{{\alpha}}^{t};\\
		\boldsymbol{p}_{j}^{t} - \theta_2 \left({1-\beta^{t}}\right) \left({\frac{1}{r-\hat{r}}- \hat{p}_{\hat{\boldsymbol{\alpha}}}^{t}}\right) - \frac{\theta_1 \beta^{t}}{r-\hat{r}} \left({\hat{r} - \hat{p}_{\hat{\boldsymbol{\alpha}}}^{t}}\right), \quad \quad  & {\alpha}_j \notin \hat{{\alpha}}^{t},
	\end{cases}
\end{equation}
where $\hat{r} = \left|{\hat{\alpha}^t} \right|$ and, 
\[
	\hat{p}_{\hat{\boldsymbol{\alpha}}}^t = \sum{\boldsymbol{p}}_i, \quad \quad \text{for} \; {\alpha}_i \in \hat{{\alpha}}^{t}.
\]

Finally, we leverage multi-task learning automata—let $\boldsymbol{\rho}_i = \boldsymbol{p}_i$ and $M=r$—to propose the RQA scheme. RQA is an iterative process that we can start it with a uniform distribution of influence factors—i.e., $\boldsymbol{\rho}= \{\frac{1}{M} \}^M$. In each iteration, the agent applies Eq. \eqref{eqn:RQA_ext_h} and submits the corresponding QMI to a quantum annealer.
After performing the necessary post-processing methods (like remediating broken-chains and applying post-quantum error correction heuristics), the agent estimates $\beta$ according to the number of satisfied constraints ($\pi_i$) and employs Eq. \eqref{eqn:la_multitask_update} to update the influence factor $\boldsymbol{\rho}$.

\subsection*{Proof of Concept: Solving SAT Instances}
For a given Boolean formula $f\left(\mathbf{x}_1, \mathbf{x}_2, \dots, \mathbf{x}_n\right)$, the problem of Boolean satisfiability (SAT) determines whether a constant replacement of values (“True” or “False”) for all Boolean variables can interpret ${f}$ as “True” \cite{biere2009handbook}. From a complexity perspective, SAT is NP-complete, and we can reduce all problems of class NP to SAT in polynomial-time \cite{garey2002computers}. The Boolean formula ${f}$ is in conjunctive normal form (CNF) if it is a conjunction (“AND”) of clauses (i.e., $f(\mathbf{x}) = C_1 \wedge C_2 \wedge \dots \wedge C_M$), where each clause is a disjunction (“OR”) of literals (a Boolean variable or its negation)—$C_i = \mathbf{l}_i \vee \mathbf{l}_j \dots \vee \mathbf{l}_k$. The maximum satisfiability problem (MAX-SAT) is an NP-hard extension of the SAT problem that aims to maximize the number of satisfying clauses \cite{biere2009handbook,ayanzadeh2019sat}.

First, we adopt the idea of penalty methods to propose a novel heuristic that maps the problem of SAT to a QMI that is executable by the D-Wave quantum processors. In other words, we aim to find coefficients of Eq. \eqref{eqn:ising_energy} such that the ground state of the resulting Ising Hamiltonian represents the satisfying solution of the original SAT instance. In this formulation, $\mathbf{z}_i$ represents the Boolean variable $\mathbf{x}_i$, and we interpret -1 and +1 as “False” and “True”, respectively. For a clause with ${k}$ literals, there are $2^k$ different possibilities among which, we can distinguish the only state that makes the clause to be false—called infeasible state. Hence, we represent each clause of the given SAT instance with two inequalities as:
\begin{equation}
	\label{eqn:cnf2ising_infeasible}
	D_i \leq E_{\mathrm{infeasible}},
\end{equation}
and,
\begin{equation}
	\label{eqn:cnf2ising_feasible}
	D_i \ge \sum{E_{\mathrm{feasible}}},
\end{equation}
where $D_i \in \mathbb{R}$ is the  boundary variable corresponding to the clause $C_i$, AND $E_{\mathrm{feasible}}$ and $E_{\mathrm{infeasible}}$ represent the contribution of the feasible and infeasible states in the ultimate energy function, respectively. For $C_i = \mathbf{x}_1 \vee \neg \mathbf{x}_4 \vee \mathbf{x}_9$, as an example, Eq. \eqref{eqn:cnf2ising_infeasible} reduces to: 
\begin{align*}
	D_i \leq &-\mathbf{h}_1 + \mathbf{h}_4 -\mathbf{h}_9 -J_{1,4} + J_{1,9} -J_{4,9},
\end{align*}
and we can represent Eq. \eqref{eqn:cnf2ising_feasible} as follows:
\begin{align*}
	D_i \ge &-\mathbf{h}_1 -\mathbf{h}_4 -\mathbf{h}_9 +J_{1,4} + J_{1,9} + J_{4,9}\\
		&-\mathbf{h}_1 -\mathbf{h}_4 + \mathbf{h}_9 +J_{1,4} -J_{1,9}-J_{4,9}\\
		&-\mathbf{h}_1 +\mathbf{h}_4 + \mathbf{h}_9 -J_{1,4} -J_{1,9} + J_{4,9}\\
		&+\mathbf{h}_1 -\mathbf{h}_4 -\mathbf{h}_9 - J_{1,4} -J_{1,9} + J_{4,9}\\
		&+\mathbf{h}_1 -\mathbf{h}_4 +\mathbf{h}_9 -J_{1,4} +J_{1,9} -J_{4,9}\\
		&+\mathbf{h}_1 + \mathbf{h}_4 -\mathbf{h}_9 +J_{1,4} -J_{1,9} -J_{4,9}\\
		&+\mathbf{h}_1 + \mathbf{h}_4 +\mathbf{h}_9 +J_{1,4} + J_{1,9} + J_{4,9}.
\end{align*}
A clause with ${k}$ literals includes $2^k - 1$ different feasible states so the size of Eq. \eqref{eqn:cnf2ising_feasible} grows exponentially with ${k}$. 
\begin{theorem}
	\label{thm:ising_sum_eng}
	Sum of the energy values for all possible states in every Ising model is zero.
\end{theorem}	
\begin{proof}
	Let ${Z}$ denotes the set of all possible states in Eq. \eqref{eqn:ising_energy}. Because spins in the Ising model (here $\mathbf{z}_I$) take their values from $\{-1,+1\}$, ${Z}$ is a closed set under the complement operation (i.e., $\mathbf{z},-\mathbf{z} \in Z$), and $|Z|=2^N$. Accordingly, sum of the energy values for all possible states in the Ising model is: 
\begin{align*}
	\sum_{\mathbf{z} \in Z}{E_{\mathrm{Ising}}{(\mathbf{z})}} &= \sum_{\mathbf{z} \in Z}{\left({\sum_{i=1}^{N}{\mathbf{h}_i\mathbf{z}_i} + \sum_{i=1}^{N}{\sum_{j=i+1}^{N}{J_{ij}\mathbf{z}_i\mathbf{z}_j}}}\right)} =\sum_{i=1}^{N}{\left({\frac{2^{N}}{2}\mathbf{h}_i + \frac{2^{N}}{2}\left(-\mathbf{h}_i\right)}\right)} + \sum_{i=1}^{N}{\sum_{j=i+1}^{N}{\left( {\frac{2^{N}}{2}J_{ij} + \frac{2^{N}}{2}\left(-J_{ij}\right)}\right)}} = 0.
\end{align*}
\end{proof}
According to Theorem \ref{thm:ising_sum_eng}, we can rewrite Eq. \eqref{eqn:cnf2ising_feasible} as follows: 
\[
	D_i \ge -E_{\mathrm{infeasible}},
\]
that obeys:
\begin{equation} 
	\label{eqn:cnf2ising_feasible_simplified}
	0 \le D_i.
\end{equation}
Note that clauses in CNF representation are connected with “AND” operator. Hence, after representing each clause of the SAT with two inequalities, Eq. \eqref{eqn:cnf2ising_infeasible} and \eqref{eqn:cnf2ising_feasible_simplified}, we aggregate the resulting sub-systems of inequalities to form a larger system of inequalities. After representing the given SAT instance with $M$ clauses with a system of $2M$ inequalities, we represent the D-Wave hardware restrictions through embedding:
\begin{equation}
\label{eqn:dwave_hrange}
	-2 \le \mathbf{h}_i \le +2,
\end{equation} 
and, 
\begin{equation}
	\label{eqn:dwave_jrange}
	-1 \le J_{ij} \le +1,
\end{equation}
where $i,j \in \{1, 2,  \dots, N\}$ and $i<j$. Finally, we solve the following objective function: 
\begin{equation}
	\label{eqn:cnf2ising_objective}
	\max_{\mathbf{h}, J, D} \; {\sum_{i=1}^{M}{D_i}},
\end{equation}
to obtain coefficients of the Ising Hamiltonian, shown in Eq. \eqref{eqn:ising_energy}, that is executable by a D-Wave QPU. Note that inequalities\eqref{eqn:cnf2ising_infeasible}, \eqref{eqn:cnf2ising_feasible_simplified}, \eqref{eqn:dwave_hrange} and \eqref{eqn:dwave_jrange} are linear. Thus, the objective function in Eq. \eqref{eqn:cnf2ising_objective} is tractable by linear programming and convex optimization techniques. Considering that biases and couplers on a D-Wave QPU are bounded, Eq. \eqref{eqn:dwave_hrange} and \eqref{eqn:dwave_jrange}, problem \eqref{eqn:cnf2ising_objective} will always converge. 

To adopt the proposed RQA scheme for addressing the NP-complete problem of SAT, we rewrite the inequality \eqref{eqn:cnf2ising_feasible_simplified} as follows:
\begin{equation}
	\label{eqn:cnf2ising_rqa}
	\boldsymbol{\rho}_i  \leq D_i,
\end{equation}
where $\boldsymbol{\rho}_i$ denotes the influence factor of the clause $C_i$. Here, we define $\boldsymbol{\rho}$ as:
\begin{equation}
	\label{eqn:cnf2ising_rqa_rho}
	\boldsymbol{\rho}_i = \frac{1}{M}-p_i,
\end{equation}
where $p_i$ is the corresponding probability of constraint $\pi_i$ (here the clause $C_i$) in Eq. \eqref{eqn:la_multitask_update}. Note that when $\boldsymbol{p}_i = \frac{1}{M}$,  the inequalities \eqref{eqn:cnf2ising_feasible_simplified} and \eqref{eqn:cnf2ising_rqa} are identical. 

The architecture of the proposed agent contains the following components: 
\begin{itemize}
	\item $\Phi^t$—set of unsatisfied clauses in episode ${t}$;
	\item $\boldsymbol{\rho}^t$—tuple of ${M}$ influence factors in episode ${t}$;
	\item $\mathrm{QMI}^t$—action of the agent in episode ${t}$, the Ising Hamiltonian for solving the given SAT instance (according to $\Phi^t$ and $\boldsymbol{\rho}^t$);
	\item $\mathbf{z}$—perception of the agent from the stochastic environment, resulting sample(s) from executing the $\mathrm{QMI}^{t}$ on a quantum annealer. 
\end{itemize}
For a given Boolean formula in CNF, the agent initializes its internal state as:
\[
	\Phi^0 = \emptyset,
\]
\[
	\boldsymbol{p}^0 = \{\frac{1}{M}\}^M.
\]
In each episode, the agent forms a system of inequalities with Eq. \eqref{eqn:cnf2ising_infeasible} and \eqref{eqn:cnf2ising_rqa}, and embeds Eq. \eqref{eqn:dwave_hrange} and \eqref{eqn:dwave_jrange}. Afterward, the agent solves problem \eqref{eqn:cnf2ising_objective}, and submits the resulting Ising Hamiltonian ($\mathrm{QMI}^t $) to a D-Wave QPU. The environment (here the D-Wave QPU) draws sample(s) from the corresponding pseudo-Boltzmann distribution, and returns the resulting sample(s).
The episode ends with updating the internal state of the agent as follows: 
\[
	\Phi^{t} = \text{set of unsatisfied clauses with } \mathbf{z}^t,
\] 
\[
	\beta^t =1 - \frac{\left|{\Phi^{t}}\right|}{M},
\]
and (finally) updating probabilities with \eqref{eqn:la_multitask_update}. In RQA, the action of the agent in episode ${t}$ depends on $\boldsymbol{\rho}_{t-1}$; therefore, Markov property holds here \cite{cox2017theory}.

\section*{Results} 
In this section, we aim to evaluate the performance of RQA scheme on solving benchmark SAT instances, and compare it with recent software and hardware enhancements to the quantum annealers. In this study, we used Z3 (from Microsoft Research) as a framework for symbolic computing implementations in Python \cite{de2008z3} and we executed each QMI on the D-Wave 2000Q quantum annealer (with lower noise), located at Burnaby, British Columbia.

For every SAT instance, we used inequalities \eqref{eqn:cnf2ising_infeasible}, \eqref{eqn:cnf2ising_rqa}, \eqref{eqn:dwave_hrange} and \eqref{eqn:dwave_jrange} to represent the given SAT instance as a system of inequalities. Afterward, we solved problem \eqref{eqn:cnf2ising_objective} for reducing the SAT to an executable QMI on a D-Wave quantum processor. Note that solving problem \eqref{eqn:cnf2ising_objective} will result in an Ising Hamiltonian which is not necessarily compatible with the D-Wave hardware graph (Chimera topology for the current generation). Therefore, we applied the minor-embedding heuristic \cite{cai2014practical} for embedding the problem to the physical lattice of qubits on a D-Wave QPU. To avoid the impact of chaining physical qubits in our evaluations, we employed fixed embeddings of cliques in all instances. 

Recent studies have revealed that using spin-reversal transforms (a.k.a. gauge transforms)—i.e., flipping the qubits randomly without altering the ground state of the original Ising Hamiltonian—can reduce analog errors of the quantum annealers \cite{pelofske2019optimizing}. Thus, as a preprocessing technique, we applied spin-reversal transforms prior to submitting the QMIs to the physical QPU. We also put a delay between measurements to reduce the sample-to-sample correlation, albeit longer run-time. 

To remediate possible broken chains in the resulting raw samples from the D-Wave QPU, we performed voting among the physical qubits of chains. After disembedding samples (i.e., representing variables in the original problem domain), we applied the multi-qubit correction (MQC) heuristic \cite {dorband2018method} which has demonstrated a significant ability to improve the probability of finding the global minimum, attained by the D-Wave QPU. Finally, we performed a local search heuristic, so-called single-qubit correction (SQC), to construct the final solution of the given SAT \cite {dorband2018method}. 

\subsection*{Experiment A: Factoring Pseudo-Prime Numbers}
In number theory, the problem of integer factoring refers to decomposing a composite integer number into the product of smaller integers, and prime factorization restricts these factors to prime numbers. Although there are debates on the class (or complexity) of this problem, there is no known efficient (non-quantum) algorithm for factoring numbers in polynomial-time \cite{balasubramanian2018integer}.

In this study, we use the problem of prime factorization as a benchmark to evaluate the performance of RQA. It is worth noting that our research objective in this paper is not to set a new record for quantum factorized integers, which for the current generation of the D-Wave quantum annealers is 1,005,973 \cite{peng2019factoring}. Indeed, since the security of the modern public-key cryptography systems (like RSA) mainly relies on the difficulty of factoring very large pseudo-prime numbers \cite{dridi2017prime,peng2019factoring}, we relied on the difficulty of prime factorization problem for generating benchmark SAT instances. 
Let $f(\mathbf{x}_1, \mathbf{x}_2)$be a Boolean function as follows:
\begin{equation}
	\label{eqn:factoring_mul}
	\mathbf{q} = f(\mathbf{x}_1, \mathbf{x}_2) = \mathbf{x}_1 \times \mathbf{x}_2,
\end{equation}
where $\mathbf{x}_1 \in \{0,1\}^{n_1}$ and $\mathbf{x}_2 \in \{0,1\}^{n_2}$ are integer-valued numbers in binary representation (here, $\mathbf{x}_1, \mathbf{x}_2 \geq2$), and the multiply operator is in binary base—each element of the vector $\mathbf{q}$ is a Boolean function of $\mathbf{x}_1$ and $\mathbf{x}_2$. Assume that $\hat{\mathbf{q}}$ is a pseudo-prime integer number in binary base (i.e., $\hat{\mathbf{q}}$ has two prime factors. We can map the problem of factoring $\hat{\mathbf{q}}$ to SAT as follows
\begin{equation}
	\label{eqn:factoring_g}
	g = f_{\mathrm{SAT}}\left({\mathbf{q}, \hat{\mathbf{q}}}\right) 
	= \bigwedge_{i=1}^{n}{\neg{\left({\mathbf{q}_i \oplus \hat{\mathbf{q}}_i}\right)}},
\end{equation}
where $n=n_1 + n_2$ denotes the length of $\mathbf{q}$. We can look at the process of generating SAT instances from a reverse-engineering viewpoint. To this end, we generated pseudo-prime numbers via multiplying two prime numbers, and represented them in binary base (denoted by $\hat{\mathbf{q}}$). For each instance, we then used the Eq. \eqref{eqn:factoring_g} to map the factorization of $\hat{\mathbf{q}}$ to a satisfiable Boolean formula. Since ${g}$ is a Boolean expression of ${x}$, we applied the Tseitin transformation \cite{li2019clausal} to represent ${g}$ in CNF. We also performed pre-processing techniques, namely  "ctx-solver-simplify", "recover-01", "propagate-values" and “reduce-args” tactics from Z3 \cite{de2013strategy}. Note that applying the Tseitin transformation can increase the size of ${g}$ linearly, due to defining auxiliary variables. Since the capacity of the current D-Wave 2000Q quantum processors is limited to a complete graph of size 63, we eliminated SAT instances (in CNF) with more than 63 Boolean variables which resulted in 136 satisfiable SAT instances.

Figure \ref{fig:RQA_factoring} illustrates results—minimum (circles), maximum (triangles), average and variance of number of unsatisfied clauses—for solving these 136 satisfiable SAT instances, and compares the performance of the proposed RQA scheme with quantum annealing (QA) and quantum annealing with multiple post-quantum processes (SMQC). To enhance the standard quantum annealing technique, we used two spin-reversal-transforms \cite{pelofske2019optimizing}, as well as the delay between measurements to reduce the inter-sample correlation.
In the second method (SMQC), we first used the multi-qubit correction (MQC) method \cite{dorband2018method}, in problem variable level—which is the state-of-the-art technique in the realm of post-quantum correction for quantum annealers—and then applied a local search to maximize the quality of results, attained by the SMQC arrangement. 

To update influence factors of clauses in RQA, Eq. \eqref{eqn:la_multitask_update}, we used $\theta_1 =0.1$ and $\theta_2 =0$. Learning automata generally require a notable number of episodes to converge to an optimal (or sub-optimal) policies. In this experiment, nevertheless, the agent terminates the process after at most $T=10$ episodes (due to QPU time limitations) or finding a solution that satisfies all clauses. Hence, we formed a hall-of-fame—a set of final solutions from all episodes—and applied MQC (followed by SQC) on them to obtain the ultimate solution of RQA. Our empirical observations showed that this technique can implicitly address the limited number of allowed episodes in RQA. Note that RQA utilizes at most same number of samples as QA and SMQC.

\begin{figure}
	\centering
	\includegraphics{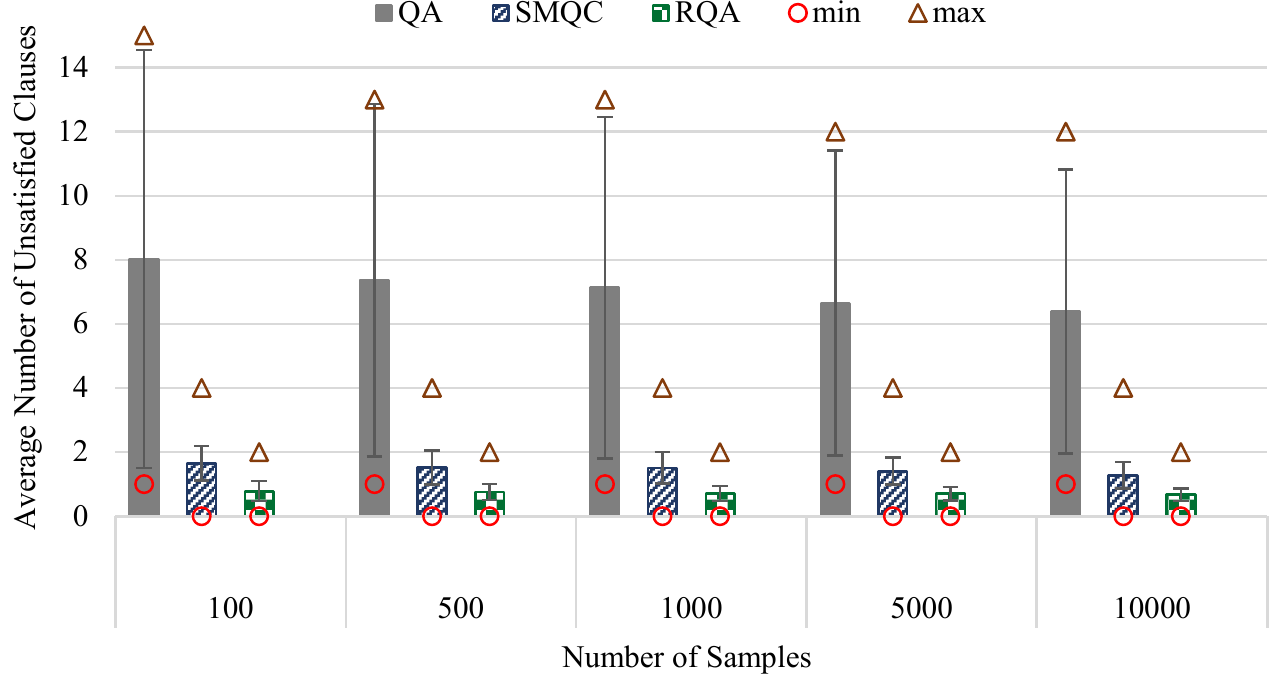}
	\caption{
Experiment results for solving 136 satisfiable SAT instances for factoring pseudo-prime numbers with quantum annealing (QA), quantum annealing with classical post-processing (SMQC) and reinforcement quantum annealing (RQA).
	}
	\label{fig:RQA_factoring}
\end{figure}

\subsection*{Experiment B: Uniform Random 3-SAT with Phase Transitions}
Sampling from the phase transition region of uniform Random 3-SAT is a common practice for generating benchmark SAT (and MAX-SAT) problems \cite{cheeseman1991really,selman1996generating,achlioptas2000generating,nudelman2004understanding}. In this experiment, as our second study case, we used the satisfiable benchmark test-set of uniform random 3-SAT with phase transitions \cite{hoos2000satlib}. Considering the capacity of the current generation of the D-Wave quantum annealers—we can embed a clique of size at most 63 on chimera architecture—we employed the test-set with 50 Boolean variables.

Figure \ref{fig:RQA_ufo} demonstrates results—minimum (circles), maximum (triangles), average and variance of number of unsatisfied clauses—for solving the first 100 instances from the benchmark test-set, and (similar to the previous experiment) compares the performance of the proposed RQA scheme with quantum annealing (QA) and quantum annealing with multiple post-quantum processes (SMQC). The setting for this experiment was identical to the previous experiment, except the number of instances (136 vs. 100) and the number of variables (variant vs. 50). 
\begin{figure}
	\centering
	\includegraphics{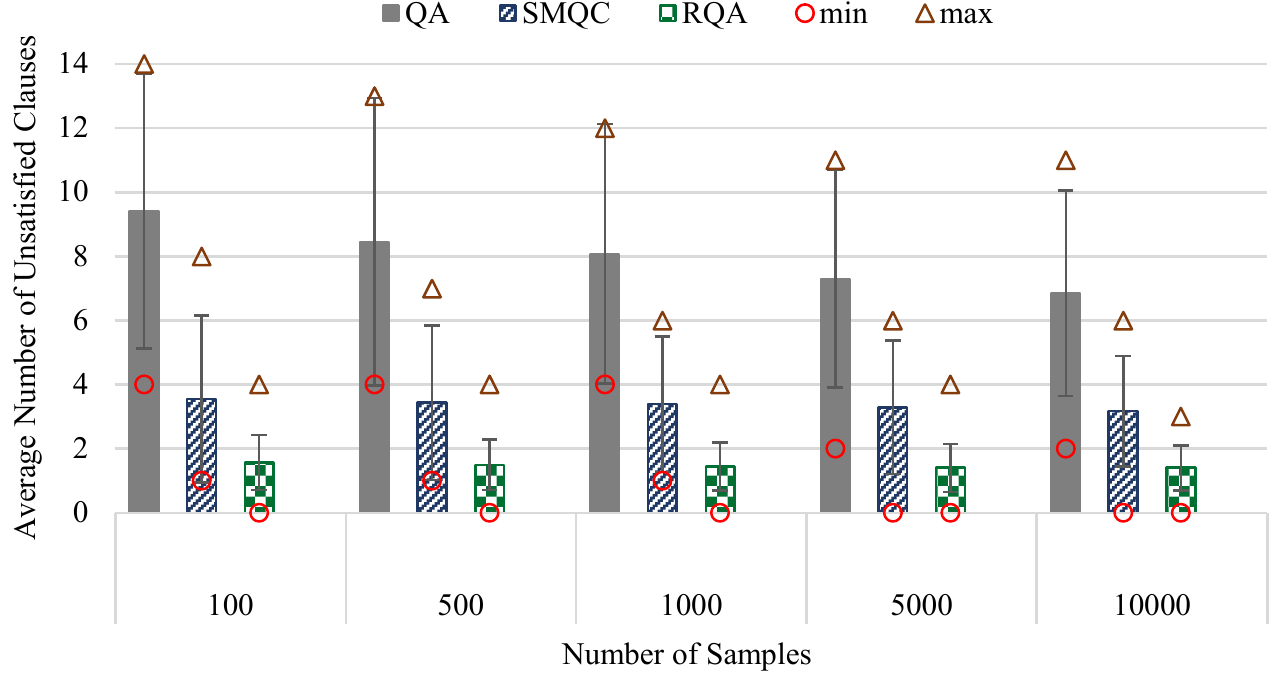}
	\caption{
Experiment results for solving 100 satisfiable uniform random 3-SAT instances with phase transitions using quantum annealing (QA), quantum annealing with classical post-processing (SMQC) and reinforcement quantum annealing (RQA).
	}
\label{fig:RQA_ufo}
\end{figure}

\subsection*{Experiment C: Run-Time Evaluation}
In this experiment, we aim to evaluate the average run-time of RQA scheme and compare it with QA and SMQC approaches. We implemented all experiments in Python 3.7.4, and executed them on a 64-bit Windows 10 based system with 32 GB RAM and Intel Xeon processor at 3.00 GHz. Figure \ref{fig:RQA_runtime} shows the average run-time of solving 100 SAT instances for QA, SMQC and RQA methods on a D-Wave 2000Q quantum processor. Note that, in this experiment, we did not include the computation time for finding the embedding of the QMI on a working graph—we used the pre-defined embeddings of cliques for the chimera architecture. For all instances, we used two spin-reversal transforms and we also enabled the inter-sample delay between samples’ reads.  
\begin{figure}
	\centering
	\includegraphics{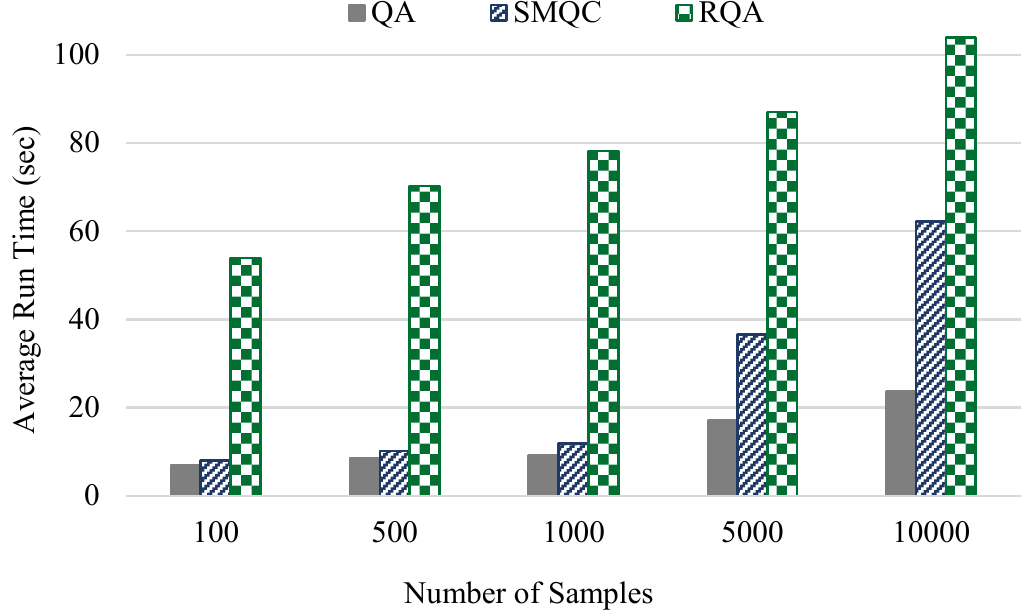}
	\caption{
	Average run-time for solving 100 SAT instances with quantum annealing (QA), quantum annealing with classical post-processing (SMQC) and reinforcement quantum annealing (RQA) approaches on a D-Wave 2000Q quantum processing unit.
}
\label{fig:RQA_runtime}
\end{figure}

\section*{Discussion}
In this study, we introduced a novel scheme—called reinforcement quantum annealing (RQA)—that leverages reinforcement learning (more specifically learning automata) to ENHANCE the quality of results, attained by the quantum annealers.
Ramezanpour (2018) has proposed to improve the simulated quantum annealing algorithm by adding reinforcement to the standard quantum annealing algorithm \cite{ramezanpour2018enhancing}. Simulated quantum annealing is an iterative algorithm (similar to simulated annealing) that is implemented and run on classical computers. On the other hand, quantum annealers are physical, single-instruction, quantum processing units where the entire annealing process is atomic. Thus, we cannot modify or adjust the annealing process after starting the annealing (i.e., putting quantum bits on their superposition). in RQA, similar to the standard reinforcement learning scheme, iterations emulate the interactions between an agent and its environment. Ramezanpour’s method, however, has an adaptive optimization scheme in which iterations simulate one quantum annealing process on a classical computer and is not applicable to physical quantum annealers.

As a proof-of-concept, we proposed a novel method for reducing SAT to minimizing Ising Hamiltonians and then demonstrated that applying the proposed RQA scheme (on a D-Wave quantum annealer) results in notably better solutions. It is worth highlighting that, however, the proposed approach (i.e., hybridization of reinforcement learning and quantum annealing) is applicable to a vast range of classic AI problems like constraint satisfaction, planning, and scheduling. 

We applied the proposed RQA scheme on two different SAT problem sets, and compared its performance with quantum annealing (QA) and quantum annealing with post-quantum error corrections (so-called SMQC) which is the state-of-the-art in the realm of quantum annealing \cite{golden2019pre,ayanzadeh2019quantum_assisted,dorband2018method}.
The first problem set includes 136 satisfiable SAT instances which represent factoring pseudo-prime numbers that have at most 63 Boolean variables, in CNF representation. Besides the length of the given composite numbers, the difficulty of integer factoring also depends on the properties of integer numbers. The hardest instances of this problem are factoring pseudo-prime numbers (product of two prime numbers) whose factors have the same size (in binary base). SAT instances in experiment A are not the hardest cases of prime-factoring— restricting the SAT instances in experiment A to a composite of the same size prime factors resulted in only 8 problems. It is worth noting that our main objective in this experiment was not to address the prime factorization nor to use quantum annealers for solving SAT or MAX-SAT problems.

Figure \ref{fig:RQA_factoring} illustrates results—minimum, maximum, average and variance of a number of unsatisfied clauses—of solving these 136 satisfiable SAT instances for 100, 500, 1,000, 5,000 and 10,000 samples. In RQA, increasing the number of samples from 100 to 10,000 reduces the average number of unsatisfied clauses from 1.57 to 1.40. Similarly, in QA and SMQC, the average number of unsatisfied clauses is reduced from 9.41 and 3.55 to 6.85 and 3.17, respectively. Although increasing the number of samples in all three methods reduces the average number of unsatisfied clauses, RQA with 100 samples outperforms both QA and SMQC approaches in all arrangements (even with 10,000 samples). It is worth highlighting that QA was not able to satisfy all clauses of any of these 136 SAT instances, even when we requested for 10,000 samples, while both SMQC and RQA methods were able to find a satisfying solution for at least one of the instances in all cases. From a robustness viewpoint, increasing the number of samples from 100 to 10,000 lowers the variance and range (the difference between maximum and minimum) of QA, SMQC and RQA from 4.28, 2.60 and 0.86 to 3.21, 1.73 and 0.70, respectively. Therefore, RQA demonstrated significantly better robustness (i.e., higher reproducibility rate), compared to both QA and SMQC approaches. 

In the second study case, shown in Fig. \ref{fig:RQA_ufo}, we used the first 100 SAT instances of the satisfiable benchmark test-set of uniform random 3-SAT with phase transitions \cite{hoos2000satlib}. Similar to the Fig. \ref{fig:RQA_factoring}, Fig. \ref{fig:RQA_ufo} demonstrates that RQA with 100 samples outperforms both QA and SMQC approaches in all settings. More specifically, increasing the number of samples from 100 to 10,000 in QA, SMQC and RQA decreases the average number of unsatisfied clauses from 8.02, 1.65 and 0.79 to 6.39, 1.28 and 0.68, respectively. Also, the variances are reduced from 6.52, 0.53 and 0.31 to 4.43, 0.42 and 0.20, respectively. Note that the minimum number of unsatisfied clauses in RQA for all settings is zero while SMQC needed at least 5,000 samples to satisfy all clauses of at least one SAT instance.

Problem-solving with a D-Wave quantum processor, in practice, requires some preprocesses (e.g., embedding and gauge transforms) and post-processes (like error correction and broken-chain remediation) that extends the total run-time from 20–2000 microseconds (pure annealing time) to some seconds and even minutes. Run-time in many pre/post-processing techniques mainly depends on the number of samples. It is a common practice in applying quantum annealers to request a few thousand (at most 10,000 per QMI) samples to increase the probability of finding the global minimum. 
Figure \ref{fig:RQA_runtime} demonstrates that increasing the number of samples increases the total run-time of RQA at a significantly lower rate (compared to QA and more specifically SMQC)—since a request of 10 times fewer samples in each iteration of RQA yet  has less pre/post-processing overhead. As an illustration, 100 samples (10 samples in each iteration of RQA), increases the total run-time ~7.6x and ~6.8x more than QA and SMQC, respectively. However, increasing the number of samples to 10,000, (1,000 samples in each iteration of RQA) only increases the total run time by ~4.4x and ~1.7x, compared to QA and SMQC, respectively. 

\section*{Acknowledgements}
This research has been supported by NASA grant (\#NNH16ZDA001N-AIST 16-0091), NIH-NIGMS Initiative for Maximizing Student Development Grant (2 R25-GM55036), and the Google Lime scholarship. We would like to thank the D-Wave Systems management team, namely Rene Copeland, for granting access to the D-Wave 2000Q quantum annealer.

\bibliography{biblio}



\end{document}